\documentclass[envcountsect]{llncs2e/llncs}
\usepackage[T1]{fontenc}
\usepackage[utf8]{inputenc}

\usepackage[ruled, vlined, linesnumbered]{algorithm2e}
\usepackage[english]{babel}
\usepackage{cancel} \usepackage{enumerate}
\usepackage[mathscr]{euscript}
\usepackage{graphicx} \usepackage[pagewise,modulo]{lineno}
\usepackage{lmodern}

\usepackage{paralist}

\usepackage{latexsym,amsmath,amssymb,textcomp,amsthm}
\usepackage{listings} \usepackage{pgf}
\usepackage{semtrans}
\usepackage{rotating}
\usepackage{subfig}

\usepackage{ stmaryrd }\usepackage{xspace} \newcommand{\macro}[2]{ \providecommand{#1}{{\ensuremath{#2}}\xspace}}
\macro{\algo}{\textsc{Algo}}

\macro{\A}{\mathcal{A}}
\macro{\cP}{\mathcal{P}}
\macro{\area}{\mathcal{A}}

\macro{\bigdot}{\bullet}
\macro{\boul}{B_{\triangle}}
\macro{\by}{\mathcal{Y}}

\macro{\CCC}{\mathcal{CIR}}
\macro{\Cn}{\mathcal{C}_{\nC}}
\macro{\cate}{\mathcal{CAT}}

\macro{\cequi}{\mathcal E}
\macro{\comp}{\texttt{t}}

\newcommand{\dd}[1]{\overline{#1}}
\macro{\dup}{dup}
\newcommand{\dest}[3]{\textsc{dest}_{#1}(#2,\lambda(#3))}
\newcommand{\destjc}[3]{\textsc{dest}_{#1}(#2,#3)}

\macro{\equi}{\mathcal{VERT}}

\macro{\eqcl}{\mathcal{C}}
\macro{\equiB}{\,\equiv_B\,}

\macro{\Gu}{\hat{G}}
\macro{\Ku}{\hat{K}}

\macro{\qcc}{\Phi}

\macro{\NT}{\mathcal{NT}}
\macro{\Neig}{N^{\triangle}}

\macro{\FNT}{\mathcal{FC}}
\macro{\F}{\mathcal{F}}

\macro{\gfam}{\mathcal F}

\macro{\G}{\mathcal{G}}
\macro{\Gport}{\mathcal{G}^{\delta}}
\macro{\Gu}{\hat{G}}

\macro{\INF}{\mathcal{IC}}
\macro{\IntCond}{\textsc{int}}

\macro{\labels}{labels}
\macro{\lab}{\lambda}
\macro{\listC}{ListComp}

 \macro{\mapCCC}{\mathcal{M_{\CCC}}}
\macro{\mapG}{\mathcal{M}}
\macro{\namemaps}{\labels}

\macro{\N}{\mathbb{N}}
\macro{\nC}{n}
\macro{\Next}{next}
\macro{\nivsup}{\mathcal{HOR}}
\macro{\nspcc}{\mathcal{I}_{ST}}
\macro{\nstV}{\mathcal{N}_{st}}
\macro{\nbfs}{\mathcal{N}_{bfs}}
\macro{\nCT}{\#CompTrav}
\macro{\no}{\textsc{no}}
\macro{\ord}{\lessdot}

\newcommand{\parag}[1]{\paragraph{\textbf{#1}.}}
\macro{\parc}{\mathcal P}
\macro{\p}{2}
\macro{\pn}{\mathcal{P}}

\newcommand{\port}[2]{\delta_{#1}(#2)}
\macro{\pred}{pred}
\macro{\rayon}{radius}

\macro{\rec}{\mathcal R}
\macro{\ringTF}{\mathcal{R}^{\geq 4}}

\macro{\SC}{\mathcal{SC}}
\macro{\Stree}{\mathcal{ST}}

\macro{\spac}{\mathbb{S}}
\macro{\ske}{\mathcal{G}^1}
\macro{\Spe}{\mathcal{S}}
\newcommand{\St}{St}

\macro{\TF}{triangle-free}
\newcommand{\tild}[1]{\widetilde{#1}}
\macro{\tu}{\tild{u}}
\macro{\tv}{\tild{v}}
\macro{\tw}{\tild{w}}
\newcommand{\tp}{\tild{p}}
\newcommand{\tq}{\tild{q}}
\newcommand{\tc}{\tild{c}}
\macro{\Topo}{\mathbb{T}}
\macro{\tree}{\mathcal{T}}
\macro{\treeC}{\mathcal{T}_{\mathcal{C}}}
\macro{\TCond}{\textsc{tri}}

\newcommand{\Un}[1]{\widehat{#1}}

\macro{\view}{\mathcal{T}}
\macro{\haut}{\shortuparrow}

\macro{\Xu}{\hat{\X}}
\macro{\X}{\mathcal{K}}

\macro{\yes}{\textsc{yes}}

\theoremstyle{plain}
\newcounter{cthm}
\newtheorem{thm}[cthm]{Theorem}
\newtheorem{lem}{Lemma}[section]

\newtheorem{prop}[lem]{Proposition}
\newtheorem{cor}[lem]{Corollary}
\newtheorem{mydef}[lem]{Definition}

\institute{LIF, Université Aix-Marseille and CNRS, FRANCE}

\title{
Anonymous Graph Exploration with Binoculars
}
\author{J\'er\'emie Chalopin, Emmanuel Godard and Antoine Naudin
}

\begin{document}

\maketitle
\pagestyle{plain}
\setcounter{page}{1}
\begin{abstract}

We investigate the exploration of networks by a mobile agent.  It is
long known that, without global information about the graph,
it is not possible to make the agent halts after the exploration
except if the graph is a tree.  We therefore endow the agent
with \emph{binoculars}, a sensing device that can show the local
structure of the environment at a constant distance of the agent
current location.

We show that, with binoculars,
it is possible to explore and halt in a large class of non-tree
networks. 
We give a complete characterization of
the class of networks that can be explored using binoculars
using standard notions of discrete topology.
Our characterization is constructive,
we present an Exploration algorithm that is universal; this algorithm
explores any network explorable with binoculars, and never halts in
non-explorable networks.
\end{abstract}


\section{Introduction}
\makeatletter{}

Mobile agents are computational units that can progress autonomously
from place to place within an environment, interacting with the
environment at each node that it is located on.  Such software robots
(sometimes called bots, or agents) are already prevalent in the
Internet, and are used for performing a variety of tasks such as
collecting information or negotiating a business deal.
More generally, when the data is physically
dispersed, it can be sometimes beneficial to move the computation to
the data, instead of moving all the data to the entity performing the
computation. The paradigm of mobile agent computing / distributed
robotics is based on this idea.  
As underlined in \cite{Das_beatcs}, the use of mobile agents 
has been advocated for numerous reasons
such as robustness against network disruptions, improving
the latency and reducing network load, providing more
autonomy and reducing the design complexity, 
and so on (see e.g. \cite{7mobile}).

For many distributed problems with mobile agents, exploring, that is
visiting every location of the whole environment, is an important
prerequisite. In its thorough exposition about Exploration by mobile
agents \cite{Das_beatcs}, S. Das presents numerous variations of the
problem. 
In particular, it can be noted that,
given some global information about the environment
(like its size or a bound on the diameter), it is always possible to
explore, even in environments where there is no
local information that enables to know, arriving on a node, whether it
has already been visited (e.g. anonymous networks).  If no global
information is given to the agent, then the only way to perform a
network traversal is to use a \emph{unlimited} traversal
(e.g. with a classical BFS or Universal Exploration Sequences
\cite{AKLLR79,uxs,R08} with increasing parameters).
This infinite process is sometimes called
\emph{Perpetual Exploration} when the agent visits
infinitely many times every node.  Perpetual Exploration has
application mainly to security and safety when the mobile agents are a
way to regularly check that the environment is safe.  But it is
important to note that in the case where no global information is
available, it is impossible to always detect when the Exploration has
been completed. This is problematic when one would like to use the
Exploration algorithm composed with another distributed algorithm.

In this note, we focus on Exploration with termination. It is known
that in general anonymous networks, the only topology that enables to
stop after the exploration is the tree-topology.  From standard
covering and lifting techniques, it is possible to see that exploring
with termination a (small) cycle would lead to halt before a complete
exploration in huge cycles.
Would it be possible to explore, with full stop, non-tree topologies
without global information? We show here that it is possible to
explore a larger set of topologies while only providing the agent with
some local information.

The information that is provided can be informally described as giving
\emph{binoculars} to the agent. This constant range sensor enables the
agent to see the relationship between its neighbours. 
Using binoculars is a quite natural enhancement
for mobile robots.
In some sense,
we are trading some a priori global information (that might be difficult to
maintain efficiently) for some local information that
the agent can \emph{autonomously} and \emph{dynamically} acquire.
We give here a complete characterization of
which networks can be explored with binoculars.

\section{Exploration with Binoculars}
\subsection{The Model}

\parag{Mobile Agents}
We use a standard model of mobile agents, that we now formally describe.
A mobile agent is a computational unit evolving in an undirected 
simple graph $G=(V,E)$ from vertex to vertex along the edges. 
A vertex can have some labels attached to it.  
There is no global guarantee on the labels, in particular
vertices have no identity (anonymous/homonymous setting),
 i.e., local labels are not guaranteed to be unique.  
The vertices are endowed with a port numbering function available to
the agent in order to let it navigate within the graph. 
Let $v$ be a vertex, we denote by
$\delta_v:V\to \N$,  the injective port numbering function giving a
locally unique identifier to the different adjacent nodes of $v$.
We denote by $\port{v}{w}$ the port number of $v$ leading to the vertex $w$,
i.e., corresponding to the edge $vw$ at $v$.
We denote by $(G,\delta)$ the graph $G$ endowed with a port numbering $\delta=\{\delta_v\}_{v\in  V(G)}$.

When exploring a network, we would like to achieve it for any port
numbering.
So we consider the set of every graph endowed with a valid port
numbering function, called $\Gport$. 
By abuse of notation, since the port numbering is usually fixed, we
denote by $G$ a graph $(G,\delta)\in\Gport$.

The behaviour of an agent is cyclic: it obtains local information
(local label and port numbers), computes some values, and moves to its
next location according to its previous computation.  We also assume
that the agent can backtrack, that is the agent knows via which port
number it accessed its current location. 
We do not assume that the starting point of the agent (that is called
the \emph{homebase}) is marked.  All nodes are a priori
indistinguishable except from the degree and the label.  We assume
that the mobile agent is a Turing machine (with unbounded local
memory).  Moreover we assume that an agent accesses its memory and
computes instructions instantaneously.  An execution $\rho$ of an
algorithm $\A$ for a mobile agent is composed by a (possibly infinite)
sequence of moves by the agent.  The length $|\rho|$ of an execution
$\rho$ is the total number of moves.

\parag{Binoculars}
Our agent can use ``binoculars'' of range 1, 
that is, it can ``see'' the graph (with port
numbers) that is induced by the
adjacent nodes of its current location. 
In order to reuse standard techniques and algorithms,
we will actually only assume that
the nodes of the graph we are exploring are labelled by this induced
balls, that is, computationally, the mobile agent has only access to
the label of its current location.

So the difference with the standard model is only in the structure of
the labels. That is the labels are
specialized and encode some local information. It is straightforward
to see that this encoding in labels is equivalent to the model with
binoculars (the ``binoculars'' 
primitive give only access to more information, it does not enable
more moves).
See Section~\ref{sect:def} for a formal definition.

\subsection{The Exploration Problem}

We consider the Exploration Problem with Binoculars for a mobile agent.
An algorithm $\A$ is an Exploration algorithm if for
any graph $G=(V,E)$ with binocular labelling, for any port numbering $\delta_G$,
starting from any arbitrary vertex $v_0\in V$, 
\begin{compactitem}
\item either the agent visits every vertex at least once and terminates;
\item either the agent never halts.
\footnote{a seemingly stronger definition could require that the agent performs
perpetual exploration in this case. It is easy to see that this is
actually equivalent for computability considerations 
since it is always possible to
compose in parallel (see below) a perpetual BFS to any never halting
algorithm.}
\end{compactitem}
The intuition for the second requirement is to model the absence of
global knowledge while maintaining safety of composition.  Since we
have no access to global information, we might not be able to visit
every node on some networks, but, in this case, we do not allow the
algorithm to appear as correct by terminating. This allows to safely
compose an Exploration algorithm with another algorithm without
additional global information.

We say that a graph $G$ is \emph{explorable} if there exists an Exploration
algorithm that halts on $G$ starting from any point.
An algorithm \A explores \gfam if it is an Exploration algorithm such
that for all $G\in\gfam$, \A explores and halts. (Note that since \A 
is an Exploration algorithm, for any $G\notin\gfam$, \A  either never
halts, or \A explores $G$.)

In the context of distributed computability, a very natural question is
to characterize the maximal sets of explorable networks.
It is not immediate that there is a maximum set of explorable networks.
Indeed, it could be possible that two graphs are explorable, but not
explorable with the same algorithm.
However, we note that explorability is monotone. That is
if $\gfam_1$ and $\gfam_2$ are both explorable then  
$\gfam_1\cup\gfam_2$ is also explorable.
Consider $\A_1$ that explores $\gfam_1$ and 
$\A_2$ that  explores $\gfam_2$ then the parallel composition of
both algorithms
(the agent performs one step of $\A_1$ then backtracks to perform
one step of $\A_2$ then backtracks, etc...; 
and when one of $\A_1$ or $\A_2$
terminates, the composed algorithm terminates) 
explores $\gfam_1\cup\gfam_2$ since 
these two algorithms guarantee to have always explored the full graph
when they terminate on any network. 
So there is actually a maximum set of explorable graphs. 

\subsection{Our Results}
In our results, we are mostly interested in computability aspects,
that is Exploration algorithms we consider could (and will reveal to)
have an unboundable complexity.  We first give a necessary condition
for a graph to be explorable with binoculars using the standard
lifting technique. Using the same technique, we give a lower bound on
the move complexity\footnote{The complexity measure we are interested
in here is the number of edge traversals (or moves) performed by the
agent during the execution of the algorithm} to explore a given
explorable graph.  Then we show that the Exploration problem admits a
universal algorithm, that is, there exists an algorithm that halts
after visiting all vertices on all explorable graphs.  This algorithm,
together with the necessary condition, proves that the explorable
graphs are exactly the graphs whose clique complexes admit a finite
universal cover (these are standard notions of discrete topology, see
Section~\ref{sect:def}). This class is larger than the class of tree
networks that are explorable without binoculars. It contains graphs
whose clique complex is simply connected (like chordal graphs or
planar triangulations), but also triangulations of the projective
plane.  Finally, we show that the move complexity of any universal
exploration algorithm cannot be upper bounded by any computable
function of the size of the network.

\parag{Related works}
To the best of our knowledge, using binoculars has never been
considered for mobile agent on graphs. When the agent can only see the
label and the degree of its current location, it is well-known that
any Exploration algorithm can only halts on trees and a standard DFS
algorithm enables to explore any tree in $O(n)$ moves. Gasieniec et
al.~ \cite{explotree} presented an algorithm that can explore any tree
with a memory of size $O(\log n)$. For general anonymous graphs,
Exploration with halt has mostly been investigated assuming at least
some global bounds, in the goal of optimizing the move complexity. It
can be done in $O(\Delta^n)$ moves using a DFS traversal while knowing
the size $n$ when the maximum degree is $\Delta$.  This can be reduced
to $O(n^3 \Delta^2 \log n)$ using Universal Exploration
Sequences \cite{AKLLR79,uxs} that are sequences of port numbers that
an agent can sequentially follow and be assured to visit any vertex of
any graph of size at most $n$ and maximum degree at most $\Delta$.
Reingold~\cite{R08} showed that universal exploration sequences can be
constructed in logarithmic space.

Trading global knowledge for structural local information by designing
specific port numberings, or specific node labels that enable easy or
fast exploration of anonymous graphs have been proposed in
\cite{CFIKP05,Gasieniec_Radzik_2008,Ilcinkas_2008}.
Note that using binoculars is a local information that can be locally
maintained contrary to the schemes proposed by these papers where the
local labels are dependent of the full graph structure.

See also \cite{Das_beatcs} for a detailed discussion about Exploration using
other mobile agent models (with pebbles for examples).

\section{Definitions and Notations} \label{sect:def}
\makeatletter{}
\subsection{Graphs} 
We always assume simple and connected graphs.  The following
definitions are standard \cite{Ros00}. Let $G$ be a graph, we denote
$V(G)$ (resp. $E(G)$) the set of vertices (resp. edges). If two
vertices $u,v\in V(G)$ are adjacent in $G$, the edge between $u$ and
$v$ is denoted by $uv$.

\parag{Paths and Cycles} A path $p$ in a graph $G$ is a sequence of
vertices $(v_0,...,v_k)$ such that $v_iv_{i+1}\in E(G)$ for every
$0\leq i < k$.  We say that the length of a path $p$, denoted by
$|p|$, is the number of edges composing it.  We denote by $p^{-1}$ the
inverted sequence of $p$.  A path is simple if for any $i\neq j$,
$v_i\neq v_j$.  A cycle is a path such that $v_0=v_k$, $k\in\N$.  A
cycle is simple if it is the empty path or the path
$(v_0,\dots,v_{k-1})$ is simple.  A \emph{loop} of length $k$ is a
sequence of vertices $(v_0,...,v_k)$ such that $v_0=v_k$ and
$v_i=v_{i+1}$ or $v_iv_{i+1}\in E(G)$, $\forall 0\leq i< k$; the
length of a loop is denoted by $|c|$.  On a graph endowed with a port
numbering, a path $p=(v_0,...,v_k)$ is labelled by
$\lambda(p)=(\delta_{v_0}(v_{1}),\delta_{v_1}(v_2),...,\delta_{v_{k-1}}(v_k))$.

The distance between two vertices $v$ and $v'$ in a graph $G$ is  
denoted by $d_G(v,v')$.
It is the length of the shortest path between $v$
and $v'$ in $G$.
 
Let $N_G(v,k)$ be the set of vertices at distance 
at most $k$ from $v$ in $G$. We denote by $N_G(v)$, the vertices at
distance at most $1$ from $v$.
We define $B_G(v,k)$ to be the subgraph of $G$ 
induced by the set of vertices $N_G(v,k)$.
In the following, we always assume that every vertex $v$ of $G$ has a
label $\nu(v)$ corresponding to the binoculars labelling of $v$.
This binoculars label corresponds to a graph $(V(\nu(v)),E(\nu(v))$
with port numbering $\tau$, that is 
isomorphic to the graph induced by $N_G(v,1)$ (with its port numbering). 
Formally, $V(\nu(v)) = N_G(v)$, $E(\nu(v))= \{ww' \mid ww'\in E(G)\}$ and
for any $ww'\in E(\nu(v))$, $\tau_w(w')=\delta_w(w')$.

\parag{Coverings}
We now present the formal definition of graph homomorphisms that capture the relation between graphs that locally look the same in our model.
A map $\varphi: V(G) \rightarrow V(H)$ from a graph $G$ to a graph $H$
is a \emph{homomorphism} from $G$ to $H$ if for every edge $uv \in
E(G)$, $\varphi(u)\varphi(v) \in E(H)$.
  A homomorphism $\varphi$ from $G$ to $H$ is a
  \emph{covering} if for every $v \in V(G)$, $\varphi_{\mid N_G(v)}$ is
  a bijection between $N_G(v)$ and $N_H(\varphi(v))$.

This standard definition is extended to labelled graphs
$(G,\delta,label)$ and $(G',$ $\delta',label')$ by adding the conditions
that $label'(\varphi(u))=label(u)$ for every $u\in V(G)$ and that
$\delta_u(v) = \delta'_{\varphi(u)}(\varphi(v))$ for every edge $uv \in
E(G)$. We have the following equivalent definition when $G$ and $G'$
are endowed with a port numbering.

\begin{prop}
\label{covering}
  Let $(G,\delta,label)$ and $(G',\delta',label')$ 
  be two labelled graphs, an homomorphism $\varphi:G\longrightarrow G'$
  is a \emph{covering} if and only if
  \begin{compactitem}
    \item for all $u\in V(G)$, $label(u)=label'(\varphi(u))$,
    \item for all $u\in V(G)$, $u$ and $\varphi(u)$ have same degree.
    \item for any $u\in V(G)$, for any $v\in
      N_G(u)$, $\delta_u(v) = \delta'_{\varphi(u)}(\varphi(v))$.
  \end{compactitem}
\end{prop}

\subsection{Simplicial Complexes}
Definitions in this section are standard from discrete topology \cite{Lyn77}. Given a set $V$,
a \emph{simplex} $s$ of dimension $n\in\N$ 
is a subset of $V$ of size $n+1$.
A \emph{simplicial complex} $K$ is a collection of
simplices such that
for every simplex $s\in K$,  $s'\subseteq s$ implies $s'\in K$.
A simplicial complex is said to be connected if the graph corresponding to the set of simplices of dimension $0$  (called vertices) 
and the set  of simplices of dimension $1$  (called edges) is connected.
We consider only connected complexes.

The \emph{star} $\St(v,K)$ of a vertex $v$ in a simplicial complex $K$
is the subcomplex defined by taking the collection of simplices of $K$
containing $v$ and their subsimplices.

It also possible to have a notion of covering for simplicial complexes.
A \emph{simplicial map} $\varphi: K \to K'$ is a map
$\varphi: V(K)\to V(K')$ such that for any simplex $s=\{v_1,...,v_k\}$ in
$K$, $\varphi(s)=\{\varphi(v_1),...,\varphi(v_k)\}$  is a simplex in $K'$.

\begin{mydef}
  A simplicial map $\varphi:K \to K'$ is a \emph{simplicial covering}
    if for every vertex $v \in V(K)$, $\varphi|_{\St(v,K)}$ is a
    bijection between $\St(v,K)$ and $\St(\varphi(v),K')$.
\end{mydef}

For any simplicial complex $K$, the following proposition shows that
there always exists a ``maximal'' covering of $K$ that is
called \emph{the universal cover} of $K$.

\begin{prop}[Universal Cover]
  For any simplicial complex $K$, there exists a possibly infinite complex
  (unique up to isomorphism) denoted $\Ku$ and a simplicial covering
  $\mu:\Ku\to K$ such that, for any complex $K'$, for any simplicial
  covering $\varphi: K'\to K$, there exists a simplicial covering
  $\gamma:\Ku\to K'$ and $\varphi\circ\gamma = \mu$.
\end{prop}

Given a graph $G=(V,E)$, the \emph{clique complex} of $G$, denoted 
$\X(G)$ is the simplicial complex formed by the cliques of $G$.
There is a strong relationship between a graph with binoculars labelling and 
its clique complex.

\begin{prop}
Let $G$ and $H$ be two graphs with binoculars labelling.  
Let $\varphi:V(G)\longrightarrow V(H)$ a map. 
The map $\varphi$ is a covering from $G$ to $H$ if and only if 
$\varphi$ is a simplicial covering from $\X(G)$ to $\X(H)$.
\end{prop}

From standard distributed computability results
\cite{YKsolvable,BVanonymous,CGM12}, it is known that the structure of
the covering maps explains what can be computed or not.
So in order to investigate the structure induced by coverings of
graphs with binoculars labelling, we will investigate the structure of
simplicial coverings of simplicial complexes. 
We will use interchangeably (if context
permits) $G$ to denote the graph with binoculars labelling or its
clique complex. We will also call simply ``coverings'' the simplicial
coverings. 

Note that not all simplicial complexes can be obtained as clique
complexes, however most of the results and vocabulary
on simplicial complexes apply. 
In particular, we define the \emph{universal cover} of a graph with
binoculars labelling to be the universal cover of its clique complex.
Note that the universal cover as a graph with
binoculars labelling can differ from its universal cover as a graph
without labels. Consider for example, the triangle network.

\parag{Homotopy}

We say that two paths $p$ and $p'$ in a complex $K$ are related by
an \emph{elementary homotopy} if one of two following conditions holds
(definitions from \cite{BH}):
\begin{compactenum}[(i)]
\item[(Backtracking)] $p'$ is obtained from $p$ by inserting or
  deleting a subpath of the form $(u,v,u)$ where $u,v$ are vertices of
  $K$.
\item[(Pushing across a 2-cell)] $p$ and $p'$ can be expressed as
   $p$ = $(v_0,\dots,v_i,q,v_j,\dots,$ $v_k)$ and $p'$ = $(v_0,\dots,v_i,q',v_j,\dots,v_k)$
  where $q$ and $q'$ are two subpaths such that the path $q^{-1}q'$ is
  a triangle of $K$.   (Note that $q$ or $q'$ may be  empty.)
\end{compactenum}

A loop $c = (u_1, \ldots, u_{i-1}, u_i, u_{i+1}, \ldots, u_k)$ is also
related by an elementary homotopy to the loop $c = (u_1, \ldots,
u_{i-1}, u_{i+1}, \ldots, u_k)$ when $u_i = u_{i+1}$.

We say that two paths (or loops) $p$ and $p'$ are homotopic equivalent
if there is a sequence of elementary homotopic paths $p_1,...,p_k$
such that $p_1=p$, $p_k=p'$ and, for every $1\leq i<k$, $p_i$ is
related to $p_{i+1}$ by an elementary homotopy.

A loop is $k-$contractible, $k\in\N$ if it can be reduced to a vertex
by a sequence of $k$ elementary homotopies.  A loop is contractible if
there exists $k\in\N$ such that it is $k-$contractible.

\parag{Simple Connectivity}
A \emph{simply connected} complex is a complex whose paths with same
endpoints are all homotopy equivalent. i.e., every cycle can be
reduced to a vertex by a finite sequence of elementary homotopies.  
These complexes have lots of interesting combinatorial and
topological properties. In the following, we rely on the fundamental result
\begin{prop}[\cite{Lyn77}]
 Let $K$ be a connected complex, then $K$ is isomorphic to its
 universal cover $\Ku$ if and only if it is simply connected.  
\end{prop}

In Figure~\ref{img:SC}, we present two examples of simplicial covering
maps, $\varphi$ is from the universal cover, and $\varphi'$ shows the general property of coverings that is that the number of vertices of the bigger complex is a multiple (here the double) of the number of vertices of the smaller complex.

\begin{figure}[t]
  \centering
  \includegraphics[height=1.66cm]{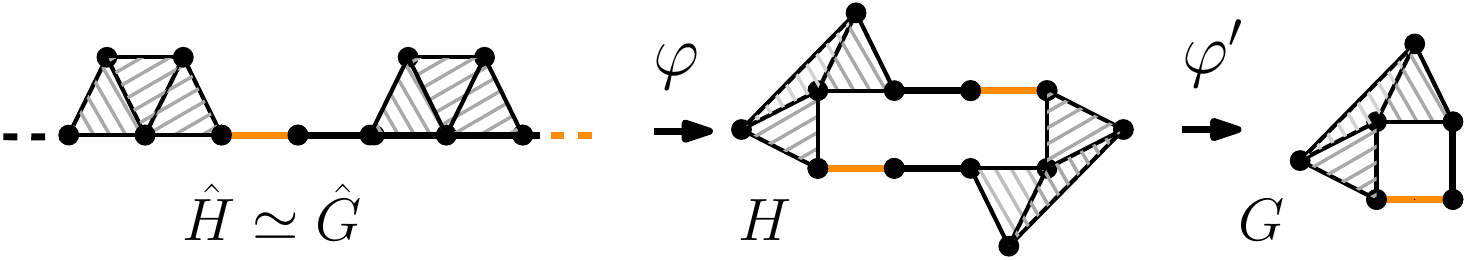}
  \caption{Simplicial Covering\label{img:SC}}
\end{figure}

\medskip

We define $\FNT=\{G\mid$ the universal cover of $\X(G)$ is finite $\}$
and $\INF=\{G \mid G$ is finite and the universal cover of $\X(G)$ is
infinite $\}$.  Note that \FNT admits one interesting sub-class
$\SC=\{G\mid G$ is finite and $\X(G)$ is simply connected$\}$.

\parag{Disk Diagrams}
Given a loop $c = (v_0, v_1, \ldots, v_n)$ in a simplicial complex
$K$, a \emph{disk diagram} $(D,f)$ of $c$ consists of a $2$-connected
planar triangulation $D$ and a simplicial map $f$ from $D$ to $K$ such
that the external face of $D$, denoted by $\partial D$, is a simple
cycle $(v_0',v_1',\ldots,v_n')$ such that $f(v'_i) =v_i$ for each
$0\leq i \leq n$. 

The \emph{area} of a disk diagram $D$, denoted $\area(D)$, is equal to
the number of faces of $D$.
A minimal disk diagram $(D,f)$ for a cycle $c$ is a disk diagram for
$c$ minimizing $\area(D)$.  The area of a cycle $c$, denoted by
$\area(c)$, is equal to the area of the minimal disk diagram for $c$.

For any $k-$contractible loop $c$, note that $\area(c) \leq k$, and
conversely if $\area(c) \leq k$, then $c$ is
$(k+|c|)$-contractible. Consequently, we have the following
alternative definition for simple connectivity.

\begin{prop}\label{prop:SC-iff-DD}
A complex $K$ is simply connected if and only if each loop $c$ of $K$
has a disk diagram. 
\end{prop}

In fact, in order to check the simple connectivity of a simplex $K$,
it is enough to check that all its simple cycles are contractible. 

\begin{prop}\label{prop:cycle_to_loop}
A complex $K$ is simply connected if and only if every simple cycle is
contractible.
\end{prop}

\begin{proof}
Consider a loop $c = (u_0, u_1, \ldots, u_k, u_0)$ that is not a
simple cycle. We prove the result by induction on the length of
$|c|$. Without loss of generality, assume that there exists $1 \leq
i \leq k$ such that $u_0 = u_i$. Let $c_1=(u_0,u_1,\ldots, u_{i-1},u_0)$
and $c_2 = (u_i, \ldots, u_n, u_i)$.  Since $|c_1| < k$ and $|c_2| <
k$, both $c_1$ and $c_2$ are contractible by the induction
hypothesis. If $1 \leq i \leq 2$, then $c$ is elementarily homotopic
to $c_2$ and we are done. Similarly, if $k-1 \leq i \leq k$, $c$ is
elementarily homotopic to $c_1$. Suppose now that $3 \leq i \leq
k-2$. We know that $c_1$ has a disk diagram $(D_1,f_1)$ such that
$\partial D_1 = (u'_0, \ldots, u'_{i-1})$ and $f_1(u'_j) = u_j$ for
all $0 \leq j \leq i-1$. Similarly, there exists a disk diagram
$(D_2,f_2)$ for $c_2$ such that $\partial D_2 = (u'_i, \ldots,
u'_{k})$ and $f_1(u'_j) = u_j$ for all $i \leq j \leq k$.  Consider
the graph $D$ obtained as follows: $V(D) = V(D_1) \cup V(D_2)$, $E(D)
= E(D_1) \cup E(D_2) \cup \{u'_0u'_i,u'_ku'_0,u'_0u'_i\}$. Note that
$D$ is a planar triangulation such that $\partial D = (u'_0, \ldots,
u'_{i-1},u'_i, \ldots, u'_k)$ (see
Figure~\ref{img:cycle-to-loop}). For any $u' \in V(D)$, let $f(u') =
f_1(u')$ if $u' \in V(D_1)$ and $f(u') = f_2(u')$ otherwise. For any
edge $u'v' \in E(D)$, either $u', v' \in V(D_1)$, or $u', v' \in
V(D_2)$, or $u'v' \in \{u'_{i-1}u'_i,u'_ku'_0,u'_0u'_i\}$. In the
first case, $f(u') = f_1(u')$ is either equal or adjacent to $f(v') =
f_1(v')$ since $f_1$ is a simplicial map. In the second case, for the
same reasons, either $f(u') = f(v')$ or $f(u')f(v') \in E(G)$. In the
last case, we know that $f(u'_i) = f(u'_0) = u_0$ is either equals or
adjacent to $u_{i-1} = f(u'_{i-1})$ (resp. to $u_{k} =
f(u'_{k})$. Consequently, $f$ is a simplicial map from $D$ to $G$ such
that $f(\partial D) = c$. Therefore $c$ has a disk diagram $(D,f)$ and
thus $c$ is contractible. 
\end{proof}

\begin{figure}[h!]
  \centering \includegraphics[height=3cm]{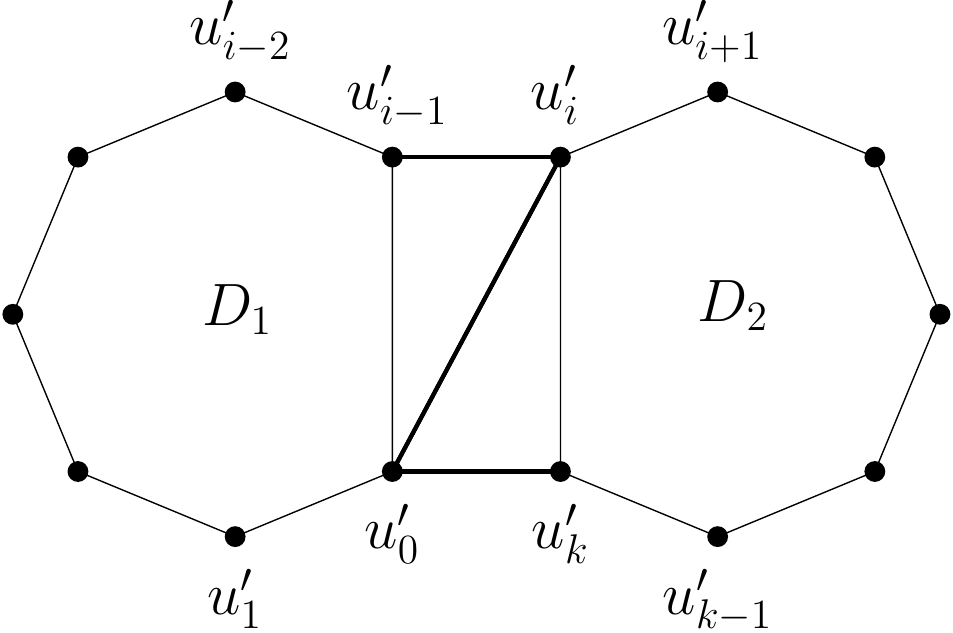} \caption{To
  the proof of
  Proposition~\ref{prop:cycle_to_loop}\label{img:cycle-to-loop}}
\end{figure}

\section{First Impossibility Result and Lower Bound}
\label{sect:NC}
First, in Lemma \ref{lemma:lifting}, we propose a Lifting Lemma for
simplicial coverings.  This lemma shows that every execution on a
graph $G$ can be lifted up to every simplicial covering $G'$ of $G$,
and in particular, to its universal cover $\Gu$.

Given an algorithm \A, and a network $G$, a starting point 
$v\in V(G)$, we denote by $(\Lambda_G^i,pos_G^i)$ the state of the
mobile agent at step $i\in\N$, where $\Lambda_G^i$ is its memory and
$pos_G^i$ is its current location in $V(G)$.

\begin{lem}
  Let $(G,\delta)$ and $(G',\delta')$ be two graphs and let $\varphi:
  G'\to G$ be a covering. Let $\A$ be an algorithm.  Let $i\in\N$, if
  $\Lambda^i_{G} = \Lambda^i_{G'}$ and $\varphi(pos^i_{G'}) = pos^i_{G}$
  then $\Lambda^{i+1}_{G} = \Lambda^{i+1}_{G'}$ and
  $\varphi(pos^{i+1}_{G'}) = pos^{i+1}_{G}$.
\end{lem}

\begin{proof}\label{proof:lifting}
The proof is standard
\cite{YKsolvable,BVanonymous,CGM12}. 
  Let $(G,\delta)$ and $(G',\delta')$ be two graphs and let $\varphi: G\to G'$ be a
  covering. Let $\A$ be an
  algorithm.
  Let $i\in\N$, assume $\Lambda^i_{G} = \Lambda^i_{G'}$ and
  $\varphi(pos^i_{G'}) = pos^i_{G}$.

  Since $\varphi$ is a covering, the graph induced by
  adjacent vertices of $pos_{G}^i$ (in $G$) is isomorphic to the graph
  induced by adjacent vertices of $pos_{G'}^i$ (in $G'$).
  Thus, $\A$ has the same input in both graphs
 and   consequently, it computes the same value and the memory is
  $\Lambda^{i+1}_G=\Lambda^{i+1}_{G'}$

  It also follows the same port number $j$, in both graphs.
  Let $v$ (resp $v'$) such that $\delta_{pos^i_G}(v)=j$ ($\delta'_{pos^i_{G'}}(v')=j$).
  Since $\varphi$ is a covering we have, by definition of a covering,
  that $\varphi(v')=v$.
\end{proof}

By iterating the previous lemma, we get the Lifting Lemma below, 
\begin{lem}[Lifting Lemma]\label{lemma:lifting} Let $\A$ be an
  Exploration algorithm and $G$ be a graph.
  For every graph $G'$ such that there exists a
  covering $\varphi:G'\to G$, for every starting points $v\in V(G)$ and 
  $u\in V(G')$ such that $v=\varphi(u)$,
  the executions of \A in $G$ and $G'$ are such that 
  $\forall i\in\N$, $\Lambda^i_G = \Lambda^i_{G'}$ and
  $\varphi(pos^i_{G'}) = pos_G^i$.
\end{lem}

Using the Lifting Lemma above, we are now able to prove a first result
about explorable graphs and the move complexity of their exploration.

\begin{prop}\label{prop:NC}
Any explorable graph $G$ belongs to $\FNT$, and the move complexity is
at least the size $|V(\Gu)|$ of its universal cover $\Gu$.
\end{prop}

\begin{proof}
Suppose it is not the case and assume there exists an exploration
algorithm $\A$ that explores a graph $G \in \INF$ when it starts from
a vertex $v_0 \in V(G)$. Let $r$ be the number of steps performed by
$\A$ on $G$ when it starts on $v_0$. 

Consider the universal cover $\Un{G}$ of the complex $G$. Consider a
covering map $\varphi$ from $\Un{G}$ to $G$ and consider a vertex
$\hat{v}_0 \in V(\Un{G})$ such that $\varphi(\hat{v}_0) = v_0$. By
Lemma~\ref{lemma:lifting}, when executed on $\Un{G}$, $\A$ stops after
$r$ steps. Consider the graph $H = B_{\Un{G}}(\hat{v}_0,r+1)$. Since
$G \in \INF$, $\Un{G}$ is infinite and $|V(H)| > r + 1$. When
executed on $H$ starting in $\hat{v}_0$, $\A$ behaves as in $\Un{G}$
during at least $r$ steps since
the first moves can only depend of $B_{\Un{G}}(\hat{v},r)$ and consequently $\A$
stops after $r$ steps when executed on $H$ starting in
$\hat{v}_0$. Since $|V(H)| > r+1$, $\A$ stops before it has visited all
nodes of $H$ and thus $\A$ is not an Exploration algorithm, a
contradiction.
The move complexity bound is obtained from the Lifting Lemma applied
to the covering map $\varphi:\Gu\to G$.
Assume we have an Exploration algorithm \A halting on $G$ at some
step $q$. If $|V(\Gu)|>q+1$
then \A halts on $\Gu$ and has not visited all vertices of $\Gu$
since at most one vertex can be visited in a step (plus the
homebase). A contradiction.
\end{proof}

This is the same lifting technique that shows that, 
without binoculars, tree networks are the only explorable
networks without global knowledge.

\section{Exploration of $\FNT$}

We propose in this section an Exploration algorithm for the \FNT
family in order to prove that this family is the maximum set of
explorable networks.  

The goal of Algorithm \ref{algo:fnt} is to visit, in a BFS
fashion, a ball  centered on the homebase of the agent until the radius 
of the ball is sufficiently large to ensure that $G$ is explored. 
Once such a radius is reached, the agent stops.
To detect when the radius is sufficiently large, we use the view of the
homebase (more details below) to search for a simply connected graph
which locally looks like the explored ball.

The view of a vertex is a standard notion 
in anonymous networks \cite{YKsolvable,BVanonymous}.
The \emph{view} of a vertex $v$ in a labelled graph $(G,label)$ is a possibly
infinite tree composed by paths starting from $v$ in $G$.

From \cite{YKsolvable}, the view $\view_G(v)$ of a vertex $v$ in $G$
is the labelled rooted tree built recursively as follows.  The root of
$\view_G(v)$, denoted by $x_0$, corresponds to $(v,label(v))$.  For
every vertex $v_i$ adjacent to $v$, we add a node $(x_i,label(v_i))$
in $\view_G(v)$ and insert an edge between $x_0$ and $x_i$ labelled by
$\delta_v(v_i)$ and $\delta_{v_i}(v)$ on the extremities corresponding
to $x_0$ and $x_i$.  To finish the construction, every node $x_i$
adjacent to $x_0$ is identified with the root of the tree
$\view_G(v_i)$.

We denote by $\view_G(v,k)$, the view $\view_G(v)$ truncated at depth
$k$. If the context permits it, we denote it by $\view(v,k)$.

Note that, in our model, $label(v)$ is actually $\nu(v)$, the graph
that is obtained using binoculars from $v$. Given an integer
$k \in \N$, we define an equivalence relation on vertices using the
views truncated at depth $k$: $v \sim_k w$ if
$\view^k(v)=\view^k(w)$.

\subsection{Presentation of the Algorithm}
Let $G$ be a complex, $v_0\in V(G)$ be the homebase of the agent in $G$
and $k$ be an integer initialized to $1$. Algorithm~\ref{algo:fnt} is divided in
phases.
At the beginning of a phase, the agent explores the ball $B(v_0,\p k)$ of
radius $\p k$ in a BFS fashion.
With this exploration, the agent records paths of length at most $2k$
originating from $v_0$ and then computes the view
$\view(v_0,\p k)$ of $v_0$.

At the end of the phase, the agent backtracks to its homebase,
and enumerates the complexes of size less than $k$
until it finds one, denoted $H$, which has a vertex whose view
at distance $2k$ is equal to $\view(v_0,\p k)$. This is the end of the phase.

If such an $H$ exists and if all its simple cycles are
$k$-contractible then we halt the Exploration.  Otherwise, $k$ is
incremented and the agent starts another phase.

\begin{algorithm}[t] 
  \LinesNumbered
  \caption{$\FNT$-Exploration algorithm}
  \label{algo:fnt}
  \BlankLine
  \Begin{
    $k=0$\;
    \Repeat{$H$ is defined and has only $k$-contractible simple cycles\label{cond2}}{
      Increment $k$ \;
      Explore at distance $2k$ and compute $\view(v_0,2k)$\; 
      Find a complex $H$ such that $|V(H)|< k$ and $\exists \tild v_0\in V(H)$
      such that  $\tild v_0\sim_{2k}v_0$\;\label{cond1}
    }
    Stop the exploration\;
  }
\end{algorithm}

Deciding the $k$-contractibility  of a given cycle is computable (by
considering all possible sequences of elementary homotopies of length
at most $k$).
Since the total number of simple cycles of a graph is finite,
Algorithm~\ref{algo:fnt} can be implemented on a Turing machine.

\subsection{Correction of the algorithm}

In order to prove the correction of this algorithm, 
we prove that when the first complex $H$ satisfying 
every condition of Algorithm \ref{algo:fnt} is found, 
then 
$H$ is actually the universal cover
of $G$.
Intuitively, this works because it is not possible to find 
a simply connected complex that looks locally the same as a
\emph{strict subpart} of another complex.

In this section, we show that if the algorithm stops when executed on
a graph $G$ from a vertex $v_0$, then the graph $H$ computed by the
algorithm is a covering of $G$ (Corollary~\ref{cor-rev}). In order to
show this, we show that if we fix a vertex $\tv_0 \in V(H)$ such that
$\tv_0 \sim_{2k} v_0$, we can define unambiguously a mapping $\varphi$
from $V(H)$ to $V(G)$ as follows: for any $\tu \in V(H)$, let $p$ be
any path from $\tv_0$ to $\tu$ in $H$ and let $u = \varphi(\tu)$ be
the vertex reached from $v_0$ in $G$ by the path labelled by
$\lambda(p)$ (Proposition~\ref{prop-chemins}). Then we show that this
mapping is a covering.

Remember that given a path $p$ in a complex $G$, $\lambda(p)$ denotes
the sequence of (outgoing) port numbers followed by $p$ in $G$. We
denote by $\dest{G}{v_0}{p}$, the vertex in $G$ reached by the
path-labelling $\lambda(p)$ from $v_0$.

We prove below a technical lemma in order to prove Proposition
\ref{prop-chemins}.

\begin{lem}\label{lem-cas-de-base}
Consider a graph $G$ such that Algorithm~\ref{algo:fnt} stops on $G$
when it starts in $v_0$.  Let $k$ and $H$ be the integer and the graph
computed by the algorithm before it stops. Consider any vertex
$\tv_0 \in V(H)$ such that $v_0 \sim_{2k} \tv_0$.

For every path $\tp=(\tv_0,\tv_1,\ldots,
\tv_{|\tp|}=\tu)$ and every cycle $\tc =
(\tu=\tu_0, \ldots, $ $\tu_{|\tc|}=u)$ in $H$ such that
  $|\tp| + |\tc| + \area(\tc) \leq 2k$, there is a
  unique vertex $u \in V(G)$ such that
  $\destjc{G}{v_0}{\lambda(\tp)}= u =
  \destjc{G}{u}{\lambda(\tc)}$.
\end{lem}

\begin{proof}
\label{proof:casdebase}
\makeatletter{}Note that $|V(H)| <k$, that $H$ is simply connected and that for every
simple cycle $\tc$ of $H$, $\area(\tc) \leq k$.  Since $v_0 \sim_{2k}
\tv_0$ and since $|\tp| \leq 2k$, there exists a unique $u \in V(G)$
such that $u = \destjc{G}{v_0}{\lambda(\tp)}$, and moreover $\nu(u) =
\nu(\tu)$. Let $m = |\tc|$; since $|\tp|+|\tc| \leq 2k$, for each $0
\leq i \leq m$, there exists $u_i$ such that
$\destjc{G}{u}{\lambda(\tu_0,\tu_1,\ldots,\tu_i)} = u_i$ and
$\destjc{G}{u_i}{\lambda(\tu_i,\tu_{i-1},\ldots,\tu_0)} = u_0$.

We prove the result by induction on $\area(\tc)$.  If $\area(\tc) =
0$, then either $|\tc| = (\tu)$ and there is nothing to prove, or $\tc
=(\tu=\tu_0,\tu_1,\tu_2=\tu)$ and $u_0 = u_{2}$ since $\nu(u) =
\nu(\tu)$. 

Suppose now that $\area(\tc) \geq 1$. Consider a minimal disk diagram
$(D,f)$ for $\tc$ and let $\dd{u} = \dd{u}_0, \dd{u}_1\ldots,
\dd{u}_{m} = \dd{u}$ be the vertices on $\partial D$ respectively
mapped to $\tu = \tu_0, \tu_1, \ldots, \tu_{m} =
\tu$. Since $D$ is a planar triangulation, there exists a path
$(\dd{u}_1=\dd{w}_1,\dd{w}_2,\ldots, \dd{w}_\ell=\dd{u}_{m-1})$ in the
neighbourhood of $\dd{u} = \dd{u}_0 = \dd{u}_m$.

We distinguish two cases, depending on whether one of the $\dd{w}_j$
is on the boundary of $\partial D$ or not.

\begin{figure}[h!]
  \centering
  \subfloat[Case 1: $\dd u_{i} = \dd w_j$]{
    \includegraphics[width=4cm]{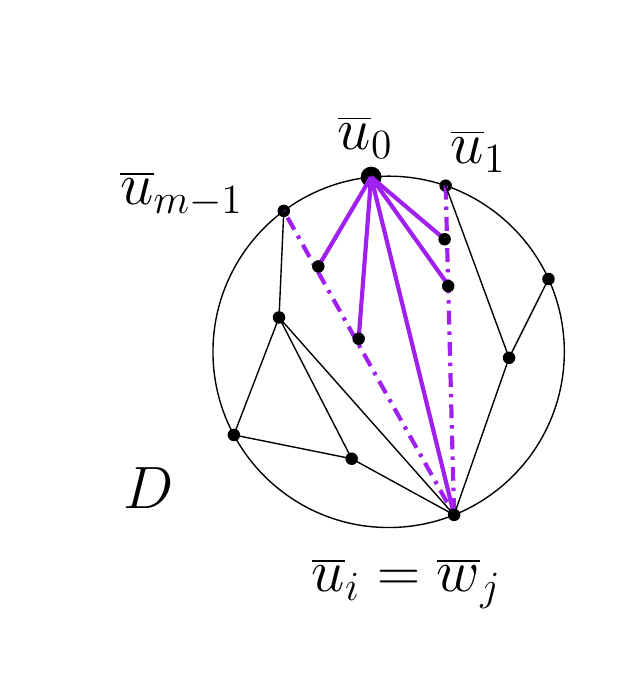}
  }
  \subfloat[Case 1.1: $f(\dd u)=f(\dd w)$]{
    \includegraphics[width=4cm]{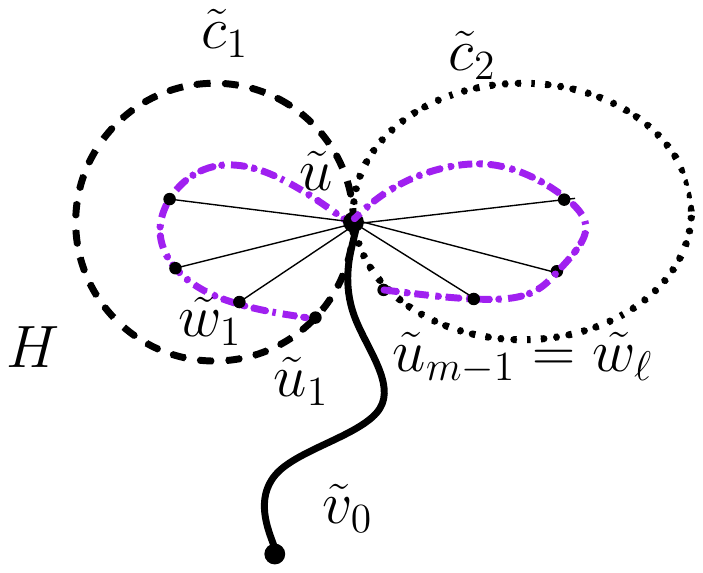}
  }
  \subfloat[Case 1.2: $f(\dd u) \neq f(\dd w)$]{
    \includegraphics[width=4cm]{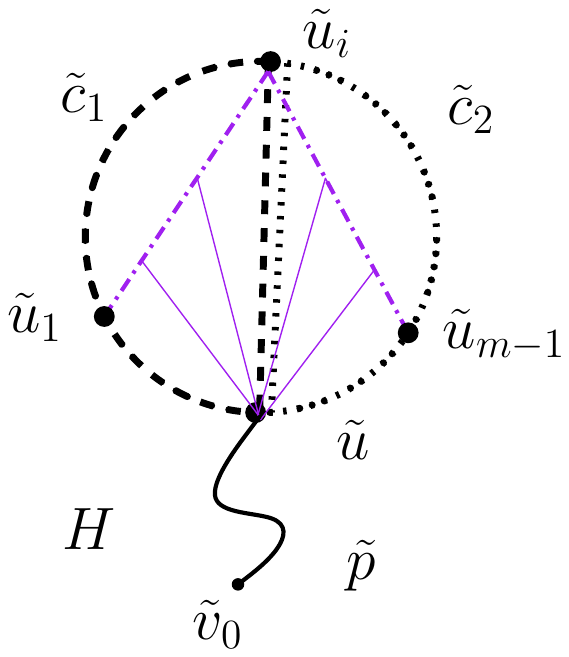}
  }
  \caption{To the proof of Lemma \ref{lem-cas-de-base}: Cas 1}
\end{figure} 

\smallskip
\noindent\textbf{Case 1.} there exists $2 \leq i \leq m-2$ and $2 \leq
j \leq \ell -1$ such that $\dd{u}_i = \dd{w}_j$.
\smallskip

Let $\dd{w} = \dd{w}_j = \dd{u}_i$ and let $\tw = f(\dd{w}) =
\tu_i$. Since $f$ is a simplicial map and since $\dd{u}\dd{w} \in
E(D)$, either $\tu = \tw$ or $\tu\tw \in
E(H)$. 

Suppose first that $\tu = \tw$ and let $\tc_1 =
(\tu=\tu_0,\tu_1,\ldots,
\tu_{i-1},\tw=\tu)$ and $\tc_2 =
(\tu=\tw,\tu_{i+1},\ldots,,\tu_m=\tu)$. For
the cycle $\tc_1$, one can construct a disk diagram $(D_1,f_1)$ from
$(D,f)$ by removing all vertices that do not lie inside the cycle
$(\dd{u}_0,\dd{u}_1,\ldots, \dd{u}_{i-1},\dd{w},\dd{u})$ and by
contracting the edge $\dd{u}\dd{w}$. Since the triangle
$\dd{u}_0\dd{w}_{\ell-1}\dd{w}_\ell$ does not appear in $D_1$,
$\area(\tc_1) \leq \area(D_1) < \area(D) = \area(\tc)$. Similarly, we have
that $\area(\tc_2) < \area(\tc)$. Moreover, note that $|\tp| +
|\tc_1| + \area(\tc_1) < |\tp| + |\tc| +
\area(\tc) \leq 2k$ and that for similar reasons, $|\tp| +
|\tc_2| + \area(\tc_2) \leq 2k$.  By induction assumption,
we have that $\destjc{G}{u}{\lambda(\tc_1)}=u$ and
$\destjc{G}{u}{\lambda(\tc_2)}=u$. Consequently,
$\destjc{G}{u}{\lambda(\tc)}=\destjc{G}{u}{\lambda(\tc_1)
  \cdot \lambda(\tc_2)} = \destjc{G}{u}{\lambda(\tc_2)} =
u$.

Suppose now that $\tu\tw\in E(H)$ and let $\tc_1 =
(\tu=\tu_0,\tu_1,\ldots, \tu_{i-1},\tw,\tu)$ and $\tc_2 =
(\tu,\tw,\tu_{i+1},\ldots,,\tu_m=\tu)$. As in the previous case, it is
easy to see that from $(D,f)$, one can construct two disk diagrams
$(D_1,f_1)$ and $(D_2,f_2)$ for $\tc_1$ and $\tc_2$ such that $\area(\tc_1)
< \area(\tc)$ and $\area(\tc_2) < \area(\tc)$. As before, $|\tp| + |\tc_1| +
\area(\tc_1) \leq 2k$ and $|\tp| + |\tc_2| + \area(\tc_2) \leq 2k$.
By induction assumption, we have that
$\destjc{G}{u}{\lambda(\tc_1)}=u$ and
$\destjc{G}{u}{\lambda(\tc_2)}=u$. Let $\tc' = \tc_1 \cdot \tc_2 =
(\tu=\tu_0,\tu_1,\ldots,\tu_{i-1},\tw,\tu,\tw,\tu_{i+1},\ldots,,\tu_m=\tu)$
and note that $\destjc{G}{u}{\lambda(\tc')} =
\destjc{G}{u}{\lambda(\tc)}$.  Consequently,
$\destjc{G}{u}{\lambda(\tc)}=\destjc{G}{u}{\lambda(\tc_1) \cdot
  \lambda(\tc_2)} =$ \linebreak $ \destjc{G}{u}{\lambda(\tc_2)} = u$ and we are
done.

\smallskip
\noindent\textbf{Case 2.} for all $1 \leq i \leq m-1$ and for all $2
\leq j \leq \ell -1$, $\dd{u}_i \neq \dd{w}_j$.
\smallskip

For each $1 \leq i \leq \ell$, let $\tw_i = f(\dd{w}_i)$.  First
suppose that $\ell=2$ and that $\tu_1=\tw_1 = \tw_2=\tu_{m-1}$. In
this case, consider the path $\tp'=\tp\cdot(u_1)$ and the cycle $\tc'
= (\tu_1, \tu_2, \ldots, \tu_{m-2})$. Note that $|\tp'| = |\tp|+1$ and
that $|\tc'|=|\tc|-2$.  From $(D,f)$, one can construct a disk diagram
for $\tc'$ by deleting the vertex $\dd{u}_0$ and by contracting the
edge $\tu_1\tu_{m-1}$. Consequently, $\area(\tc') < \area(\tc)$ and
$|\tp'|+|\tc'|+\area(\tc') < |\tp|+|\tc|+\area(\tc) \leq 2k$. By the
induction assumption,
$\destjc{G}{u_1}{\lambda(\tc')}=u_1$. Consequently,
$\destjc{G}{u}{\lambda(\tc)}=\destjc{G}{u_1}{\lambda(\tc')\cdot
  \lambda(\tu_{m-1},\tu)} = \destjc{G}{u_1}{\lambda(\tu_{1},\tu)} =
u$.

\begin{figure}[h!]
  \centering
  \subfloat{
    \includegraphics[width=3.5cm]{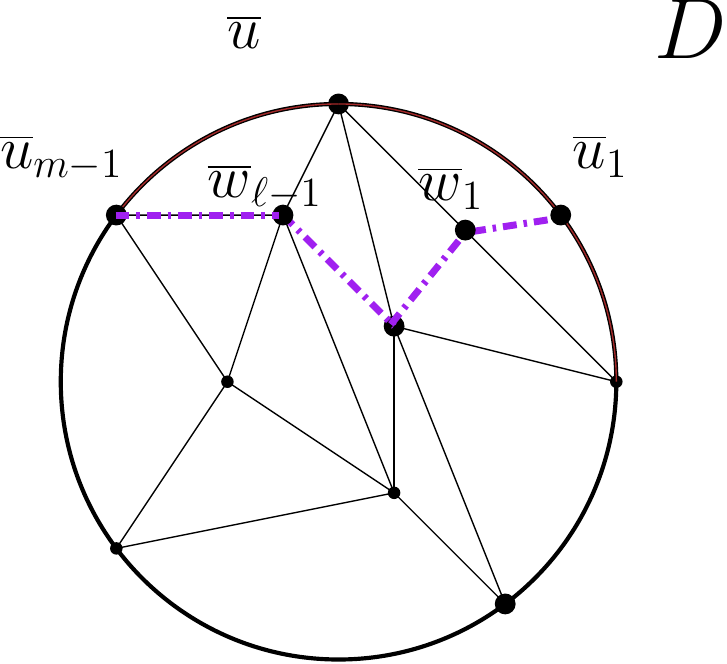}
  }
  \subfloat{
    \includegraphics[width=5cm]{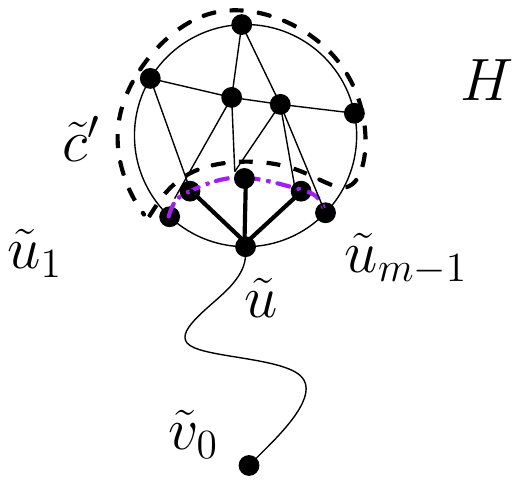}
  }
  \caption{To the proof of Lemma \ref{lem-cas-de-base}: Cas 2}
\end{figure} 
Suppose now that $\tw_1 \neq \tw_\ell$ or that $\ell > 2$.  Since
$(D,f)$ is a minimal disk diagram, we have that for each $1 \leq i
\leq \ell-1$, $\tw_i\neq \tw_{i+1}$ and $\tw_i\tw_{i+1} \in E(H)$ and
$\tu\tw_i \in E(H)$.

Let $\tp'=\tp\cdot(\tu_1)$, let $\tp_1 = (\tu_{1}, \tu_{2}, \ldots,
\tu_{m-1})$, let $\tp_2 =
(\tu_{m-1}=\tw_{\ell}, \tw_{\ell-1}, \ldots,$ $
\tw_{1}=\tu_{1})$ and let $\tc'=\tp_1 \cdot \tp_2$.  From $(D,f)$, one
can construct a disk diagram for $\tc'$ by deleting the vertex
$\dd{u}_0$, and consequently, $\area(\tc') < \area(\tc)$.  Moreover,
note that $|\tp'| = |\tp| + 1$, that $|\tc'| = |\tc| + \ell - 3$ and
$\area(\tc') \leq \area(\tc) - \ell +1$. Consequently, $|\tp'| +
|\tc'| + \area(\tc') \leq |\tp|+|\tc|+\area(\tc) -1 \leq 2k$.
Consequently, by induction,
$\destjc{G}{u_{1}}{\lambda(\tc')}=\destjc{G}{u_{m-1}}{\lambda(\tp_2)}=u_{1}$.

Let $\tp_2' = (\tu_{1}=\tw_{1}, \tw_{2}, \ldots,
\tw_{\ell}=\tu_{m-1})$ and note that
$\destjc{G}{u_{m-1}}{\lambda(\tp_2)\cdot\lambda(\tp_2')}=
u_{m-1}$. Consequently, $\destjc{G}{u}{\lambda(\tc)}=
\destjc{G}{u_1}{\lambda(\tp_1)\cdot\lambda(\tu_{m-1},\tu)}
=\destjc{G}{u_1}{\lambda(\tp_1)\cdot\lambda(\tp_2)
  \cdot\lambda(\tp_2')\cdot\lambda(\tu_{m-1},\tu)} =
\destjc{G}{u_1}{\lambda(\tp_2')\cdot\lambda(\tu_{m-1},\tu))}$.  Since
$v_0 \sim_{2k} \tv_0$ and since $|\tp|< 2k$,
$\nu(\tu)=\nu(u)$. Consequently, since $\tp_2' = (\tw_1, \ldots,
\tw_\ell)$ is a path lying in the neighbourhood of $\tu$,
$\destjc{G}{u_1}{\lambda(\tp_2')\cdot\lambda(u_{m-1},u)} =
\destjc{G}{u}{\lambda(\tu,\tu_1)\cdot\lambda(\tp_2')\cdot\lambda(\tu_{m-1},\tu)}
= u$ and consequently, $\destjc{G}{u}{\lambda(\tc)}= u$.

\end{proof}

\begin{prop}\label{prop-chemins}
Consider a graph $G$ such that Algorithm~\ref{algo:fnt} stops on $G$
when it starts in $v_0$.  Let $k\in\N$ and let $H$ be the graph
computed by the algorithm before it stops. Consider any vertex $\tv_0
\in V(H)$ such that $v_0 \sim_{2k} \tv_0$.

For any vertex $\tu \in V(H)$, for any two paths
$\tq,\tq'$ from $\tv_0$ to $\tu$ in $H$,
$\destjc{G}{v_0}{\lambda(\tq)} = \destjc{G}{v_0}{\lambda(\tq')}$.
\end{prop}

\begin{proof}
\label{proof:casgeneral}
\makeatletter{}  Suppose the result is not true and consider the set $\cP_0$ of all
  couples of paths $(\tq,\tq')$ that are counterexamples to the result
  such that $|\tq| \leq |\tq'|$. Among all couples of paths in
  $\cP_0$, let $\cP_1$ be the set of all couples of paths such that
  $|\tq|+|\tq'|$ is of minimum length. Among all couples of paths in
  $\cP_1$, let $\cP_2$ be the set of all couples of paths that
  minimizes $|\tq|-\ell$ where $\ell$ is the length of the longest common
  prefix of $\tq$ and $\tq'$.

  Let $(\tq,\tq') \in \cP_2$ and let $\tq = (\tv_0=\tu_0,
  \tu_1,\ldots, \tu_m=\tu)$ and $\tq' = (\tv_0=\tu_0', \tu_1',\ldots,
  \tu'_{m'}=\tu)$.

  Suppose first that the path $\tq$ is not simple, i.e., there exists
  $i < j$ such that $\tu_i=\tu_j$. Choose $i$ and $j$ such that $j$ is
  minimum. Consequently, for all $0 \leq i_1 < i_2 \leq j-1$,
  $\tu_{i_1}\neq \tu_{i_2}$. Consider the path $\tp = (\tu_0, \tu_1,
  \ldots, \tu_i)$ and the cycle $\tc = (\tu_i, \tu_{i+1}, \ldots,
  \tu_{j}= \tu_i)$. Note that since all vertices of $\tp$ and $\tc$
  are distinct and since $|V(H)| \leq k$, $|\tp| + |\tc| \leq
  k$. Since all simple cycles of $H$ are $k$-contractible, $\area(\tc)
  \leq k$. Consequently, by Lemma~\ref{lem-cas-de-base},
  $\destjc{G}{v_0}{\lambda(\tu_0,\ldots,\tu_i)} =
  \destjc{G}{v_0}{\lambda(\tu_0,\ldots,\tu_j)}$. Let $\tq_1 =
  (\tu_0,\ldots,\tu_{i-1},\tu_{i}= \tu_j, \tu_{j+1},\ldots,\tu_m)$ and
  note that from what we just prove, we have that
  $\destjc{G}{v_0}{\lambda(\tq)} =
  \destjc{G}{v_0}{\lambda(\tq_1)}$. Since $\tq_1$ is a path from
  $\tv_{0}$ to $\tu$ and since $|\tq_1| < |\tq|$,   from the definition of $\cP_1$, $\destjc{G}{v_0}{\lambda(\tq_1)} =
  \destjc{G}{v_0}{\lambda(\tq')}$ and consequently,
  $\destjc{G}{v_0}{\lambda(\tq)} = \destjc{G}{v_0}{\lambda(\tq')}$ and
  thus, $\tq$ is a simple path.  Using the same arguments, we can show
  that $\tq'$ is also a simple path.

  Let $\ell$ be the length of the longest common prefix of $\tq$ and
  $\tq'$. Note that if $\ell = |\tq|$, then $\tq = \tq'$ and there is
  nothing to prove. Suppose that $\ell < |\tq|$ and note that for all
  $i \leq \ell$, $\tu_i = \tu'_i$.  Consider the smallest index $i >
  \ell$ such that $\tu_i$ appears in $\tq'$. Let $j$ be the smallest
  index such that $\tu_i = \tu'_j$. Note that since $\tq$ and $\tq'$
  are simple, $j > \ell$ and all vertices $u_1, \ldots, u_\ell$,
  $u_{\ell+1}, \ldots, u_i$, and $u'_{\ell+1}, \ldots, u'_{j-1}$ are
  distinct.

  Let $\tq_1 = (\tu_0,\tu_1, \ldots, \tu_i)$ and let $\tq'_1 =
  (\tu'_0,\tu'_1, \ldots, \tu'_j)$ Consider the simple path $\tp =
  (\tu_0,\tu_1,\ldots,\tu_\ell)$ and the simple cycle $\tc = (\tu_\ell,
  \tu_{\ell+1}, \ldots, \tu_i=\tu'_j, \tu'_{j-1}, \ldots, \tu'_\ell)$. Since
  all vertices from $\tp$ and $\tc$ are distinct, $|\tp| + |\tc|\leq
  k$, and since all simple cycles of $H$ are $k$-contractible,
  $\area(\tc) \leq k$. Consequently, by Lemma~\ref{lem-cas-de-base},
  there exists a unique vertex $w \in V(G)$ such that $w =
  \destjc{G}{v_0}{\lambda(\tq_1)} = \destjc{G}{v_0}{\lambda(\tq_1')}$.

  Let $\tq_2 = (\tu_i,\tu_{i+1},\ldots,\tu_m)$ and $\tq'_2 =
  (\tu'_j,\tu'_{j+1},\ldots,\tu'_{m'})$. Note that
  $\destjc{G}{v_0}{\lambda(\tq)} = \destjc{G}{w}{\lambda(\tq_2)}$ and
  that $\destjc{G}{v_0}{\lambda(\tq')} =
  \destjc{G}{w}{\lambda(\tq'_2)}$.  If $j < i$, let $\tq'_3 =
  \tq'=\tq'_1\cdot\tq'_2$ and let $\tq_3= \tq'_1\cdot\tq_2$.  If $i
  \leq j$, let $\tq_3 = \tq =\tq_1\cdot\tq_2$ and let $\tq'_3 =
  \tq_1\cdot\tq'_2$. Note that if $j < i$ or if $i < j$, then $|\tq_3|
  + |\tq_3'| < |\tq| + |\tq'|$. If $i = j$, $|\tq_3| + |\tq'_3|= |\tq|
  + |\tq'|$, and the length of the common prefix of $\tq_3$ and
  $\tq'_3$ is $i > \ell$.  Consequently, in any case, from our choice
  of $\tq$ and $\tq'$, we know that $\destjc{G}{v_0}{\lambda(\tq_3)} =
  \destjc{G}{v_0}{\lambda(\tq_3')}$. Moreover, since $w =
  \destjc{G}{v_0}{\lambda(\tq_1)} = \destjc{G}{v_0}{\lambda(\tq_1')}$,
  we have that $\destjc{G}{w}{\lambda(\tq_2)} =
  \destjc{G}{v_0}{\lambda(\tq_3)} = \destjc{G}{v_0}{\lambda(\tq'_3)} =
  \destjc{G}{w}{\lambda(\tq'_2)}$. Consequently,
  $\destjc{G}{v_0}{\lambda(\tq)} = \destjc{G}{w}{\lambda(\tq_2)} =
  \destjc{G}{w}{\lambda(\tq'_2)} = \destjc{G}{v_0}{\lambda(\tq')}$,
  contradicting our assumption on $\tq, \tq'$.
 
\end{proof}
From Proposition \ref{prop-chemins} above we can define an
homomorphism between $G$ and the graph $H$ computed during the execution.
We prove below that, in fact, the homomorphism is a simplicial covering. 
\begin{cor}\label{cor-rev}
Consider a graph $G$ such that Algorithm~\ref{algo:fnt} stops on $G$
when it starts in $v_0 \in V(G)$ and let $H$ be the graph computed by
the algorithm before it stops.  The complex $H$ is the universal cover
of $G$.
\end{cor}

\begin{proof}
  By the definition of Algorithm~\ref{algo:fnt}, the complex $H$ is
  simply connected. Consequently, we just have to show that $H$ is a
  covering of $G$.
  
  Consider any vertex $\tv_0 \in V(H)$ such that $v_0 \sim_{2k}
  \tv_0$. For any vertex $\tu \in V(H)$, consider any path $\tp_{\tu}$
  from $\tv_0$ to $\tu$ and let $\varphi(\tu)
  =\destjc{G}{v_0}{\lambda(\tp_{\tu})}$. From
  Proposition~\ref{prop-chemins}, $\varphi(\tu)$ is independent from
  our choice of $\tp_{\tu}$. Since $v_0 \sim_{2k} \tv_0$ and since
  $|V(H)| \leq k$, for any $\tu \in V(H)$, $\nu(\varphi(\tu)) =
  \nu(\tu)$. Consequently, for any $\tu \in V(H)$ and for any neighbour
  $\tw \in N_H(\tu)$, there exists a unique $w \in N_G(\varphi(\tu))$
  such that $\lambda(\tu,\tw)=\lambda(\varphi(\tu),w)$. Conversely, for
  any $w\in N_G(\varphi(\tu))$, there exists a unique $\tw \in
  N_H(\tu)$ such that $\lambda(\tu,\tw)=\lambda(\varphi(u),w)$. In
  both cases, since $\tp_\tw=\tp_{\tu}\cdot(\tu,\tw)$ is a path from
  $\tv_0$ to $\tw$, by Proposition~\ref{prop-chemins}, $\varphi(\tw) =
  \destjc{G}{v_0}{\lambda(\tp_\tw)}= \destjc{G}{u}{\lambda(\tu,\tw)}=
  w$. Consequently, $\varphi$ is a covering from $H$ to $G$ that
  preserves the binoculars labelling. Therefore, the complex $H$ is a
  covering of the complex $G$.
\end{proof}

To finish to prove that Algorithm \ref{algo:fnt} is an $\FNT$
Exploration algorithm, we remark 
that, with connected graphs, coverings are always
surjective, therefore $G$ has been explored when the algorithm stops.

\begin{thm}\label{thm:explo_algo}
  Algorithm \ref{algo:fnt} is an Exploration algorithm
  for $\FNT$.
\end{thm}
\begin{proof}
  From Corollary~\ref{cor-rev}, we know that if
  Algorithm \ref{algo:fnt} stops, then the graph $H$ computed by the
  algorithm is a covering of $G$. Moreover, since $|V(G)| \leq
  |V(H)| \leq k$ and since the agent has constructed $\view_G(v,k)$,
  it has visited all vertices of $G$.

  In order to show the theorem, we just have to prove that
  Algorithm \ref{algo:fnt} always halts on any graph
  $G \in \FNT$. Since $G \in \FNT$, its universal cover $\Gu$ is
  finite, its number of simple cycles is finite and there exists
  $q\in\N$ such that every simple cycle of $\Gu$ is $q-$contractible.
  Without loss of generality, assume that $|V(\Gu)| \leq
  q$. Consequently, if when starting on $v_0$, the algorithm does not
  halt, there exists an iteration of the main loop of
  Algorithm~\ref{algo:fnt} such that $k \geq q$. Since $\Gu$ is the
  universal cover of $G$, there exists $\tv_0 \in V(\Gu)$ such that
  $\view_G(v_0) = \view_\Gu(\tv_0)$. Consequently, $\Gu$ is a complex
  such that $\view_G(v_0,k) = \view_\Gu(\tv_0,k)$, $|V(\Gu)| \leq k$,
  and all its simple cycles are $k$-contractible. Therefore,
  Algorithm~\ref{algo:fnt} halts, a contradiction.
\end{proof}

From Proposition~\ref{prop:NC} and Theorem~\ref{thm:explo_algo} above, we
get the immediate corollary
\begin{corollary}
 The family $\FNT$ is the maximum set of Explorable networks.
\end{corollary}

\section{Complexity of the Exploration Problem}

\makeatletter{}In the previous section, we did not provide any bound on the number of
moves performed by an agent executing our universal exploration
algorithm. In this section, we study the complexity of the problem and
we show that there does not exist any exploration algorithm for all
graphs in $\FNT$ such that one can bound the number of moves performed
by the agent by a computable function in the size of the graph.

The first reason that such a bound cannot exist is rather simple: if
the universal cover $\Un{G}$ of a complex $G$ is finite, then by
Lemma~\ref{lemma:lifting}, when executed on $G$, any exploration
algorithm has to perform at least $|V(\Un{G})|$ steps before it
halts. In other words, one can only hope to bound the number of moves
performed by an exploration algorithm on a graph $G$ by a function of
the size of the universal cover of the complex $G$.

However, in the following theorem, we show that even if we consider
only graphs with simply connected clique complexes (i.e., they are
isomorphic to their universal covers), there is no
Exploration algorithm for this class of graph such that one can bound
its complexity by a computable function. Our proof relies on a result
of Haken~\cite{Haken73} that show that it is undecidable to detect
whether a finite simplicial complex is simply connected or
not.

\begin{theorem}\label{th:unbounded}
Consider any algorithm $\A$ that explores every finite graph $G \in
\SC$. For any computable function $\comp : \N \rightarrow \N$, there
exists a graph $G \in \SC$ such that when executed on $G$, $\A$
executes strictly more than $\comp(|V(G)|)$ steps.
\end{theorem}

\begin{proof}
\label{proof:unbounded}
\makeatletter{}  Suppose this is not true and consider an algorithm $\A$ and a
  computable function $\comp : \N \rightarrow \N$ such that for any
  graph $G \in \SC$, $\A$ visits all the vertices of $G$ and stops in
  at most $\comp(|V(G)|)$ steps.  We show that in this case, it is
  possible to algorithmically decide whether the clique complex of any
  given graph $G$ is simply connected or not. However, this problem is
  undecidable~\cite{Haken73} and thus we get a
  contradiction\footnote{Note that the original result of
    Haken~\cite{Haken73} does not assume that the simplicial complexes
    are clique complexes. However, for any simplicial complex $K$, the
    barycentric subdivision $K'$ of $K$ is a clique complex that is simply
    connected if and only if $K$ is simply connected
    (see~\cite{Hatcher}).}.

  Algorithm~\ref{algo-simple-connectivity} is an algorithm that takes
  as an input a graph $G$ and then simulates $\A$ on $G$ for
  $\comp(|V(G)|)$ steps. If $\A$ does not stop within these
  $\comp(|V(G)|)$ steps, then by our assumption on $\A$, we know that
  $G \notin \SC$ and the algorithm returns \textsc{no}. If $\A$ stops
  within these $\comp(|V(G)|)$ steps, then we check whether there
  exists a graph $H$ with at most $\comp(|V(G)|)$ vertices such that
  the complex $H$ is a strict covering of the complex $G$. If such an
  $H$ exists, then $G \notin \SC$ and the algorithms returns
  \textsc{no}. If we do not find such an $H$, the algorithm returns
  \textsc{yes}.

  \begin{algorithm}
  \LinesNumbered
  \caption{An algorithm to check simple connectivity}
  \label{algo-simple-connectivity}
  \KwIn{a graph $G$ }
  \BlankLine
    Simulate $\A$ starting from an arbitrary starting vertex $v_0$
    during $\comp(|V(G)|)$ steps  \;   
    \label{line:simulation}
    \eIf{$\A$ halts within $\comp(|V(G)|)$ steps}
        {\eIf{there exists a graph $H$ such that $|V(G)| < |V(H)| \leq
            \comp(|V(G)|)$ and such that the complex $H$ is a covering
            of the complex $G$}
          { \Return \textsc{no}; // the complex $G$ is not simply connected}
          { \Return \textsc{yes}; // the complex $G$ is  simply connected}
        }
        { \Return \textsc{no}; // the complex $G$ is not simply connected}
  \end{algorithm}

  In order to show our algorithm is correct, it is sufficient to show
  that when the algorithm returns \textsc{yes} on a graph $G$, the
  complex $G$ is simply connected. Suppose it is not the case and let
  $\Un{G}$ be the universal cover of the complex $G$.  Consider a
  covering map $\varphi$ from $\Un{G}$ to $G$ and let $\hat{v}_0 \in
  V(\Un{G})$ be any vertex such that $\varphi(\hat{v}_0) = v_0$.  By
  Lemma~\ref{lemma:lifting} and Proposition~\ref{prop:NC}, when executed on
  $\Un{G}$ starting in $\hat{v}_0$, $\A$ stops after at most
  $\comp(|V(G)|)$ steps.

  If $\Un{G}$ is finite, then $\Un{G} \in \SC$ and by our assumption
  on $\A$, $\A$ must explore all vertices of $G$ before it
  halts. Consequently, $\Un{G}$ is a covering of the complex $G$ with
  at most $\comp(|V(G)|)$ vertices; in this case, the algorithm
  returns \textsc{no} and we are done.

  Assume now that $\Un{G}$ is infinite. Let $r = \comp(|V(G)|)$ and
  let $B = B_{\Un{G}}(\hat{v}_0,r)$. Note that when $\A$ is executed on
  $\Un{G}$ starting in $\hat{v}_0$, $\A$ does not visit any node that
  is not in $B$. Given two vertices, $\hat{u}, \hat{v} \in V(\Un{G})$,
  we say that $\hat{ u} \equiB \hat{v}$ if there exists a path from
  $\hat{u}$ to $\hat{v}$ in $\Un{G} \setminus B$. It is easy to see
  that $\equiB$ is an equivalence relation, and that every vertex of
  $B$ is the only vertex in its equivalence class.  For a vertex
  $\hat{u} \in V(\Un{G})$, we denote its equivalence class by
  $[\hat{u}]$. Let $H$ be the graph defined by $V(H) = \{[\hat{u}]
  \mid \hat{u} \in V(\Un{G})\}$ and $E(H) = \{[\hat{u}][\hat{v}] \mid
  \exists \hat{u}' \in [\hat{u}], \hat{v}' \in [\hat{v}],
  \hat{u}'\hat{v}' \in E(\Un{G})\}$.

  We now show that the complex $H$ is simply connected. Let $\varphi :
  V(\Un{G}) \rightarrow V(H)$ be the map defined by $\varphi(\hat{u})
  = [\hat{u}]$. By the definition of $H$, for any edge
  $\hat{u}\hat{v} \in E(\Un{G})$, either $[\hat{u}]=[\hat{v}]$, or
  $[\hat{u}][\hat{v}] \in E(H)$. Consequently, $\varphi$ is a
  simplicial map.  Consider a cycle $c = ({u}_1, {u}_2, \ldots,
  {u}_p)$ in $H$. By the definition of $H$, there exists a cycle $c' =
  (\hat{u}_{1,1}, \ldots, \hat{u}_{1,\ell_1}, \hat{u}_{2,1}, \ldots, \hat{u}_{2,\ell_2}, \ldots, \hat{u}_{p,1}, \ldots, \hat{u}_{p,\ell_p})$
  in $G$ such that for each $1 \leq i \leq p$ and each $1 \leq
  j \leq \ell_i$, $\varphi(\hat{u}_{i,j}) = {u}_i$. Since $\Un{G}$ is
  simply connected, there exists a disk diagram $(D,f)$ such that
  $f(\partial D) = c'$. Consequently, $(D,\varphi \circ f)$ is a disk
  diagram for the loop $\varphi(c')$ that is homotopic to
  $c$. Consequently, $c$ is contractible and from
  Proposition \ref{prop:SC-iff-DD},thus $H$ is simply connected.

  Since $G$ is finite, the degree of every vertex of $\Un{G}$ is
  bounded by $|V(G)|$ and consequently, the number of equivalence
  classes for the relation $\equiB$ is finite. Consequently, the graph
  $H$ is finite and thus $H \in \SC$. Moreover, since for every
  $\hat{u} \in B$, $[\hat{u}] = \{\hat{u}\}$, the ball
  $B_H([\hat{v}_0],r)$ is isomorphic to $B$. Consequently, when $\A$
  is executed on $H$ starting in $[\hat{v}_0]$, $\A$ stops after at
  most $r$ steps before it has visited all vertices of $H$,
  contradicting our assumption on $\A$.

\end{proof}

\section{Conclusion}

Enhancing a mobile agent with binoculars, we have shown that, even
without any global information it is possible to explore and halt in
the class of graphs whose clique complex have a finite universal
cover. This class is maximal and is the counterpart of tree networks
in the classical case without binoculars.  Note that, contrary to the
classical case, where the detection of unvisited nodes is somehow
trivial (any node that is visited while not backtracking is new, and
the end of discovery of new nodes is immediate at leaves), we had here
to introduced tools from discrete topology in order to be able to
detect when it is no more possible to encounter ``new'' nodes.

The class where we are able to explore is fairly large and has been
proved maximal when using binoculars of radius 1.  But note that for
triangle-free networks, using binoculars does not change anything.
More generally, from the proof techniques in Section~\ref{sect:NC}, it
can also be shown that providing only local information (e.g. with
binoculars of higher range) cannot be enough to explore all graphs
(e.g. graphs with large girth). 

While providing binoculars is a natural enhancement, it appears here
that explorability increases at the cost of a huge increase in
complexity, that cannot be expected to be reduced for fundamental
Turing computability reasons for all explorable graphs.  But
preliminary results show that it is possible to explore with
binoculars with a linear move complexity in a class that is way larger
that the tree networks.  So the fact that the full class of explorable
networks is not explorable efficiently should not hide the fact that
the improvement is real for large classes of graphs.  One of the
interesting open problem is to describe the class of networks for
which explorability is increased while still having reasonable move
complexity, like networks that are explorable in linear time.

Note that our Exploration algorithm can actually compute the universal
cover of the graph, and therefore yields a Map Construction algorithm
if we know that the underlying graph has a simply connected clique
complex. However, note that there is no algorithm that can construct
the map for all graphs of \FNT. Indeed, there exist graphs in \FNT
that are not simply connected (e.g. triangulations of the projective
plane) and by the Lifting Lemma, they are indistinguishable from their
universal cover. Note that without binoculars, the class of trees is
not only the class of graphs that are explorable without information,
but also the class of graphs where we can reconstruct the map without
information.  Here, adding binoculars, not only enables to explore
more networks but also give a model with another computability
structure : some problems (like Exploration and Map Construction) are
no longer equivalent in the binocular model.

\vspace{-0.5cm}
\footnotesize
\newcommand{\etalchar}[1]{$^{#1}$}

\end{document}